\documentclass[12pt,3p,number]{elsarticle}

\usepackage{amssymb}
\usepackage{graphicx}
\usepackage{url}
\usepackage{floatrow}
\usepackage{placeins}

\usepackage{xspace}
\usepackage{tikz}
\usepackage[noend,boxed,linesnumbered]{algorithm2e}
\usetikzlibrary{arrows,snakes,shapes}
\tikzset{node/.style={draw,circle,fill=none,red,gray,thick,inner sep = 1.6pt}}

\newcommand{\G}{{\cal{G}}}
\newcommand{\IG}[2]{{\G^{#1}[#2]}} 
\newcommand{\llad}[2]{{\mathcal{L}^{#1}[#2]}} 
\newcommand{\rlad}[2]{{\mathcal{R}^{#1}[#2]}} 
\newcommand{\tinter}{{\scshape $T$-Interval-Connectivity}\xspace}
\newcommand{\inter}{{\scshape Interval-Connectivity}\xspace}

\newcommand{\tstable}{{\scshape $T$-Stability}\xspace}
\newcommand{\stable}{{\scshape Stability}\xspace}

\newcommand{\boldpara}[1]{\par\noindent{{\bf {#1}}\,}}

\definecolor{ao}{rgb}{0.0, 0.5, 0.0}

\newcommand{\joe}[1]{#1}
\newcommand{\ir}{j}
\newcommand{\review}[1]{#1}

\newcommand{\il}{j-1}

\newtheorem{theorem}{Theorem}
\newtheorem{lemma}{Lemma}

\newtheorem{observation}{Observation}

\newdefinition{definition}{Definition}
\newproof{proof}{Proof}

\begin{document}

\begin{frontmatter}

\title{Efficiently Testing $T$-Interval Connectivity in Dynamic Graphs\tnoteref{label1}}
\tnotetext[label1]
{A short version of this paper appeared in the proceedings of CIAC 2015~\cite{CKNP15}.\\
\indent
Part of this work was done while Joseph Peters was visiting the LaBRI as a guest professor of the University of Bordeaux. This work was partially funded by the ANR projects DISPLEXITY (ANR-11-BS02-014) and  ESTATE (ANR-16-CE25-0009-03). This study has been carried out in the frame of ``The Investments for the Future'' Programme IdEx Bordeaux – CPU (ANR-10-IDEX-03-02). The work of Joseph Peters was partially supported by NSERC of Canada.}

\author[LABRI]{Arnaud Casteigts}
\ead{arnaud.casteigts@labri.fr}

\author[LABRI]{Ralf Klasing}
\ead{ralf.klasing@labri.fr}

\author[LABRI]{Yessin M.~Neggaz}
\ead{mneggaz@labri.fr}

\author[SFU]{Joseph G.~Peters}
\ead{peters@cs.sfu.ca}

\address[LABRI]{LaBRI, CNRS, University of Bordeaux, France}
\address[SFU]{School of Computing Science, Simon Fraser University, Burnaby, BC, Canada}


\begin{abstract}
Many types of dynamic networks are made up of durable entities whose links evolve over time. When considered from a {\em global} and {\em discrete} standpoint, these networks are often modelled as evolving graphs, i.e. a sequence of graphs $\G=(G_1,G_2,...,G_{\delta})$ such that $G_i=(V,E_i)$ represents the network topology at time step $i$. Such a sequence is said to be $T$-interval connected if for any $t\in [1, \delta-T+1]$ all graphs in $\{G_t,G_{t+1},...,G_{t+T-1}\}$ share a common connected spanning subgraph. In this paper, we consider the problem of deciding whether a given sequence $\G$ is $T$-interval connected for a given $T$. We also consider the related problem of finding the largest $T$ for which a given $\G$ is $T$-interval connected. We assume that the changes between two consecutive graphs are arbitrary, and that two operations, {\em binary intersection}
and {\em connectivity testing},
are available to solve the problems. We show that $\Omega(\delta)$ such operations are required to solve both problems, and we present optimal $O(\delta)$ online algorithms for both problems.
We extend our online algorithms to a dynamic setting in which connectivity is based on the recent evolution of the network.
\end{abstract}

\end{frontmatter}

\section{Introduction}

Dynamic networks consist of entities making contact over time with one
another. The types of dynamics resulting from these interactions are
varied in scale and nature. For instance, some of these networks remain
connected at all times~\cite{OW05}; others are always
disconnected~\cite{JFP04} but still offer some kind of connectivity over
time and space ({\em temporal} connectivity); others are recurrently
connected, periodic, etc. 
\review{All of these contexts can be represented as
properties of dynamic graphs (also called time-varying graphs, evolving graphs, or temporal graphs). A number of such classes were identified in recent literature and organized into a hierarchy in~\cite{CFQS12}. Each of these classes corresponds to specific properties which play a role either in the complexity or in the feasibility of distributed problems. For example, it was shown in~\cite{CFMS15} that if the edges are {\it recurrent} (i.e. if an edge appears once, it will reappear infinitely often), denoted class ${\cal R}$, then such a property guarantees the feasibility of a certain type of optimal broadcast with termination detection (namely, {\it foremost} broadcast). However, it is not sufficient to satisfy other measures of optimality, such as {\it shortest} or {\it fastest} broadcast. Strenghtening the assumption to having a {\it bound} on the reappearance time (class ${\cal B}$) makes it possible to satisfy the shortest measure, and making it even stronger by assuming that edges are periodic (class ${\cal P}$) enables fastest broadcast. These three classes have been shown to play a role in a variety of problems (see e.g.~\cite{AKM14,FMS13,RSCW14}). Another important class, which is less constrained (and thus more general) is the class of all graphs with recurrent temporal connectivity (i.e. all nodes can recurrently reach each other through journeys), corresponding to class ${\cal C}_5$ in the hierarchy of~\cite{CFQS12}. This property is very general, and it is used (implicitly or explicitly) in a number of recent studies addressing distributed problems in highly-dynamic environments~\cite{BDD16,BT16,BDK16,DKP15}. Interestingly, this property was considered more than two decades ago by Awerbuch and Even~\cite{AE84}.
}

\review{Given a dynamic graph, a natural question to ask is to which of the classes
this graph belongs. This question is interesting in several respects. Firstly, most of the known
classes correspond to necessary or sufficient conditions
for given distributed problems or algorithms (broadcast, election, spanning
trees, token forwarding, etc.). Thus, being able to classify a graph in the
hierarchy is useful for determining which problems can be solved
on that graph. Furthermore, it is useful for
choosing a good algorithm in settings where some properties are guaranteed (as in the above example with classes ${\cal R}, {\cal B}$, and ${\cal P}$). Hence, when targeting a given scenario from the real world, 
an algorithm designer may first record 
some topological traces from the target environment and then test which useful properties are satisfied.
Alternatively, online algorithms that process dynamic graphs as they evolve
could accomplish the same goal without the need to store all the traces.
Besides distributed algorithms, a growing amount of research is now focusing on testing properties (or computing structures) in dynamic graphs. Recent examples include computing reachability graphs~\cite{BCCJN14,WDCG12}, enumerating maximal cliques~\cite{VLM16}, and determining the hardness of computing even simple metrics like temporal diameter (that is, how long it takes in the worst case to communicate through journeys) when the evolution is not known in advance~\cite{CG13}. A less recent but seminal example is the computation of foremost, shortest, or fastest journeys~\cite{BFJ03} which can be used indirectly to test membership in a number of dynamic graph classes~\cite{CCF09}.}

\review{In this paper, we focus on another important class of dynamic graphs called {\em $T$-interval connected} graphs. This class captures two fundamental aspects of a network---{\em stability} and {\em connectivity}---through a single parameter $T$. This parameter was 
identified in~\cite{KLO10} as playing a role in several distributed
problems, such as determining the size of a network or computing a
function of the initial inputs of the nodes. 
The definition of this class is closely related to a certain way of representing dynamic networks. It is often
convenient, when looking at the evolution of the topology from a {\em global} standpoint, to represent a dynamic graph as a sequence of
graphs $\G=(G_1,G_2,...,G_\delta)$, each of which corresponds to
the state of the topology at a given (discrete) time instant.
This representation, first suggested by Harary and Gupta~\cite{HG97}, was called {\em untimed} evolving graphs in~\cite{BFJ03} (in which a more general version is considered).
}
\joe{Informally, $T$-interval connectivity requires that, for every $T$ consecutive graphs in the sequence $\G$, there exists a common connected spanning subgraph.
In \cite{KLO10}, the authors focussed on two problems. In the {\em $k$-token dissemination problem}, each of $k$ nodes has a piece of information that must be collected by all $n$ nodes of the graph. The authors were especially interested in the case $k=n$ (also known as {\em gossiping}). To solve the $n$-token dissemination problem, the nodes need to solve the {\em counting problem} which determines the number of nodes in the graph, so that the nodes know how many pieces of information to collect. Solving the $n$-token dissemination problem allows the computation of any function of the initial states of the nodes. It was shown that both problems can be solved in $O(n^2)$ rounds in $1$-interval connected graphs. For $T$-interval connected graphs with $T > 1$ and $T$ unknown, they showed that both problems can be solved in $O(\min\{n^2, n + n^2 \log(n)/T\})$ rounds. If $T$ is known, then both problems are solvable in $O(n+n^2/T)$ rounds.}

\joe{The class of $T$-interval connected graphs} generalizes the class of dynamic 
graphs that are connected at all time instants~\cite{OW05}. Indeed, the
latter corresponds to the case that $T=1$. From a
set-theoretic viewpoint, however, every $T>1$ induces a class of graphs
that is a strict subset of the class in~\cite{OW05} because a graph that is
$T$-interval connected is obviously $1$-interval connected. Hence,
$T$-interval connectivity is more specialized in this sense.

In this paper, we look at the problem of deciding whether a given
sequence $\G$ is $T$-interval connected for a given $T$. We also
consider the related problem of finding the largest $T$ for which the
given $\G$ is $T$-interval connected. 
\joe{We adopt a model that allows us to focus on high-level strategies without being concerned about the lower-level details of specific data structures for representing the graphs. The basic elements in our model are graphs and the basic operations on these elements are {\em binary intersection} (given two graphs, compute their intersection) and {\em connectivity testing} (given a graph, decide whether it is connected). These two high-level operations have a strong and natural connection with the problems that we are studying. This approach is suitable for general dynamic graphs in which the details of changes between successive graphs in a sequence are arbitrary. We discuss our model in more detail in Section~\ref{sec:definitions}.
}

We first show that
both problems require $\Omega(\delta)$ such operations using the \joe{straightforward}
argument that every graph of the sequence must be considered at least
once. \review{Interestingly,} 
we show that both problems can be solved
using only $O(\delta)$ such operations and we develop optimal online
algorithms that achieve these matching bounds.
\joe{Hence, our efficient high-level implementations counterbalance the costs of implementing the operations. In fact, both operations can be computed in linear time in the number of edges using suitable lower-level data structures (see Observation~\ref{obs:cost}) and could benefit from dedicated circuits (or optimized code).}

The paper is organized as follows. Section~\ref{sec:definitions}
presents the main definitions and makes some basic observations,
including the fact that both problems can be solved using
$O(\delta^2)$ operations (intersections or connectivity tests) by a
naive strategy that examines $O(\delta^2)$ intermediate graphs.
Section~\ref{sec:intermediate} presents a second strategy, yielding
upper bounds of $O(\delta \log \delta)$ operations for both problems.
Its main interest is in the fact that it can be parallelized, and this
allows us to classify both problems as being in {\bf NC} (i.e. Nick's
class). In Section~\ref{sec:optimal} we present an optimal
strategy which we use to solve both problems online in $O(\delta)$
operations. This strategy exploits structural properties of the
problems to construct carefully selected subsequences of the
intermediate graphs. In particular, only $O(\delta)$ of the
$O(\delta^2)$ intermediate graphs are selected for evaluation by the
algorithms.
In Section~\ref{sec:dynamic}, we extend our online algorithms to a
dynamic setting in which the measure of connectivity is based on the recent
evolution of the network.

\section{Definitions and Basic Observations}
\label{sec:definitions}

\boldpara{Graph Model.}
In this work, we consider dynamic graphs that are given as untimed evolving graphs, that is, a sequence $\G=(G_1,G_2,...,G_\delta)$ of graphs such that $G_i=(V,E_i)$ describes the network topology at (discrete) time $i$. The parameter $\delta$ is called the {\em length} of the sequence $\G$ (also known as the {\em lifetime}). It corresponds to the number of time steps that this graph covers. Observe that $V$ is non-varying; only the set of edges varies. 
Unless otherwise stated, we consider {\em undirected} edges throughout the paper, which is the setting in which $T$-interval connectivity was originally introduced. However, the fact that our algorithms \review{use high-level operations} allows them to work exactly the same for $T$-interval strong connectivity (which is the analogue of $T$-interval connectivity for directed graphs~\cite{KLO10}), provided that both basic operations (i.e. intersection and connectivity test) are given. As we shall discuss, these operations have linear cost in the number of edges in both directed and undirected graphs.

\begin{definition}[Intersection graph]
Given a (finite) set $S$ of graphs $\{G'=(V,E'),G''=(V,E''),\dots\}$, we call the graph $(V,\cap\{E',E'',\dots\})$ the {\em intersection graph} of $S$ and denote it by $\cap\{G',G'',\dots\}$. When the set consists of only two graphs, we talk about {\em binary intersection} and use the infix notation $G' \cap G''$. If the intersection involves a consecutive subsequence $(G_i,G_{i+1},\dots,G_j)$ of a dynamic graph $\G$, then we denote the intersection graph $\cap\{G_i,G_{i+1},\dots,G_j\}$ simply as $G_{(i,j)}$.
\end{definition}

\begin{definition}[\boldmath $T$-interval connectivity]
\review{A dynamic graph $\G$ is said to be {\em $T$-interval connected}, for $T \leq \delta$, if 
the intersection graph $G_{(t,t+T-1)}$ is connected for every $t\in [1, \delta-T+1]$. In other words, all graphs in $\{G_t,G_{t+1},...,G_{t+T-1}\}$ share a common connected spanning subgraph \joe{for every $t\in [1, \delta-T+1]$.}
}
\end{definition}

\begin{definition}[Testing \boldmath $T$-interval connectivity]
We will use the term \tinter to refer to the problem of deciding whether a dynamic graph $\G$ is $T$-interval connected for a given $T$.
\end{definition}

\begin{definition}[Interval connectivity]
We will use \inter to refer to the problem of finding $\max\{T: \G$ is $T$-interval connected$\}$ for a given $\G$.
\end{definition}

Let $\G^{T}=(G_{(1,T)},G_{(2,T+1)},...,G_{(\delta-T+1,\delta)})$. We call $\G^T$ the $T^{th}$ {\em row} in $\G$'s intersection hierarchy, as depicted in Figure~\ref{fig:example}. A particular case is $\G^1=\G$. For any $1 \leq i \leq \delta-T+1$, we define $\IG{T}{i} = G_{(i,i+T-1)}$. We call $\IG{T}{i}$ the {\em $i^{th}$ element of row $\G^T$} and $i$ is called the {\em index of $\IG{T}{i}$ in row $\G^T$}. 

\begin{observation}\label{obs:row-connected}
  \review{By definition,} a dynamic graph $\G$ is $T$-interval connected if and only if all graphs in $\G^T$ are connected. 
\end{observation}


\def\carre (#1,#2,#3){
  \path (#1,#2+1) node[node] (a){};
  \path (#1+1,#2+1) node[node] (b){};
  \path (#1,#2) node[node] (c){};
  \path (#1+1,#2) node[node] (d){};
  \path (#1+.5,#2-.8) node (e){#3};
}
\newcommand{\g}[2]{\ensuremath{G_{#1}^{#2}}}
\tikzstyle{every node}=[]
\begin{figure}
  \centering
  \begin{tabular}{cccc}
    \begin{tikzpicture}[xscale=.5, yscale=.5]
		
		\node (G) [] 
    at (-2.9,0.5) {\large{$\G=\G^1$}};
    \node (G) [] 
    at (-0.4,3.4) {\large{$\G^2$}};
    \node (G) [] 
    at (1,6.3) {\large{$\G^3$}};
    \node (G) [] 
    at (2.4,9.2) {\large{$\G^4$}};
    
      \carre (0,0,$G_1$)
      \draw (a)--(b);
      \draw (a)--(c);
      \draw (c)--(b);
      \draw (c)--(d);
      \draw (b)--(d);
      \carre (3,0,$G_2$)
      \draw (a)--(c);
      \draw (c)--(b);
      \draw (c)--(d);
      \draw (a)--(b);
      \carre (6,0,$G_3$)
      \draw (a)--(c);
      \draw (d)--(b);
      \draw (c)--(d);
      \draw (a)--(b);
      \carre (9,0,$G_4$)
      \draw (a)--(c);
      \draw (b)--(c);
      \draw (d)--(b);
      \draw (c)--(d);
      \draw (a)--(b);
      \carre (12,0,$G_5$)
      \draw (b)--(c);
      \draw (a)--(d);
      \draw (d)--(b);
      \draw (c)--(d);
      \draw (a)--(b);
      \carre (15,0,$G_6$)
      \draw (b)--(c);
      \draw (d)--(b);
      \draw (c)--(d);
      \draw (a)--(b);
      \carre (18,0,$G_7$)
      \draw (b)--(c);
      \draw (c)--(d);
      \draw (a)--(b);
      \carre (21,0,$G_8$)
      \draw (a)--(b);
      \draw (a)--(c);
      \draw (b)--(c);
      \draw (c)--(d);

      \carre (1.5,2.9,$G_{(1,2)}$)
      \draw (a)--(b);
      \draw (a)--(c);
      \draw (c)--(b);
      \draw (c)--(d);
      \carre (4.5,2.9,$G_{(2,3)}$)
      \draw (a)--(c);
      \draw (a)--(b);
      \draw (c)--(d);
      \carre (7.5,2.9,$G_{(3,4)}$)
      \draw (a)--(c);
      \draw (b)--(d);
      \draw (a)--(b);
      \draw (c)--(d);
      \carre (10.5,2.9,$G_{(4,5)}$)
      \draw (b)--(c);
      \draw (b)--(d);
      \draw (a)--(b);
      \draw (c)--(d);
      \carre (13.5,2.9,$G_{(5,6)}$)
      \draw (b)--(c);
      \draw (b)--(d);
      \draw (a)--(b);
      \draw (c)--(d);
      \carre (16.5,2.9,$G_{(6,7)}$)
      \draw (b)--(c);
      \draw (c)--(d);
      \draw (a)--(b);
      \carre (19.5,2.9,$G_{(7,8)}$)
      \draw (a)--(b);
      \draw (b)--(c);
      \draw (c)--(d);

      \carre (3,5.8,$G_{(1,3)}$)
      \draw (a)--(b);
      \draw (a)--(c);
      \draw (c)--(d);
      \carre (6,5.8,$G_{(2,4)}$)
      \draw (a)--(b);
      \draw (a)--(c);
      \draw (c)--(d);
      \carre (9,5.8,$G_{(3,5)}$)
      \draw (a)--(b);
      \draw (b)--(d);
      \draw (c)--(d);
      \carre (12,5.8,$G_{(4,6)}$)
      \draw (b)--(c);
      \draw (b)--(d);
      \draw (a)--(b);
      \draw (c)--(d);
      \carre (15,5.8,$G_{(5,7)}$)
      \draw (b)--(c);
      \draw (a)--(b);
      \draw (c)--(d);
      \carre (18,5.8,$G_{(6,8)}$)
      \draw (b)--(c);
      \draw (a)--(b);
      \draw (c)--(d);

      \carre (4.5,8.7,$G_{(1,4)}$)
      \draw (a)--(b);
      \draw (a)--(c);
      \draw (c)--(d);
      \carre (7.5,8.7,{$G_{(2,5)}$})
      \draw[ultra thick] (a)--(b);
      \draw[ultra thick] (c)--(d);
      \carre (10.5,8.7,$G_{(3,6)}$)
      \draw (a)--(b);
      \draw (c)--(d);
      \draw (b)--(d);
      \carre (13.5,8.7,$G_{(4,7)}$)
      \draw (a)--(b);
      \draw (c)--(d);
      \draw (b)--(c);
      \carre (16.5,8.7,$G_{(5,8)}$)
      \draw (a)--(b);
      \draw (c)--(d);
      \draw (b)--(c);
    \end{tikzpicture}
  \end{tabular}
  \caption{\label{fig:example} \review{An example} of an intersection hierarchy for a given dynamic graph $\G$ of length $\delta=8$. Here, $\G$ is 3-interval connected, but not 4-interval connected; $\G^4$ contains a disconnected graph $G_{(2,5)}$ because $G_2,G_3,G_4,G_5$ share no connected spanning subgraph.}
\end{figure}
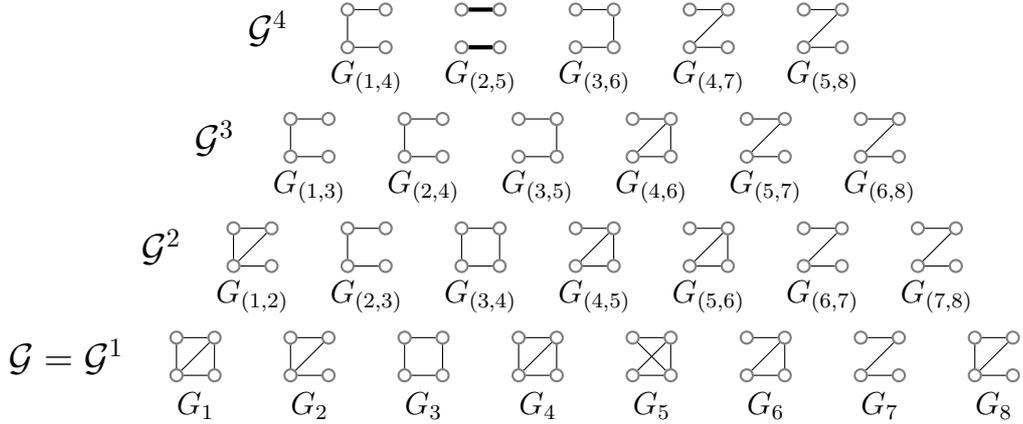

\boldpara{Computational Model.}
As \review{mentioned} in Observation~1, the concept of $T$-interval connectivity can be reformulated quite naturally in terms of the connectivity of some intersection graphs. For this reason, we consider two building block operations: {\em binary intersection} (given two graphs, compute their intersection) and {\em connectivity testing} (given a graph, decide whether it is connected). \review{This approach is suitable for general dynamic graphs \joe{in which} the details of changes between successive graphs in a sequence are arbitrary. If more structural information about the evolution of the dynamic graphs were known,
for example, if it were known that the number of changes between each pair of consecutive graphs is bounded by a constant,
then algorithms could benefit from the use of sophisticated data structures and a lower-level approach \joe{than ours} might be more appropriate.}

\begin{observation}[Cost of the operations]
\label{obs:cost}
\joe{Using a sorted adjacency list in which the neighbours of each node are sorted}, a binary intersection can be performed in linear time in the number of edges. Checking connectivity of a graph can also be done in linear time in the number of edges. In the case of undirected graphs, it can be done by building a depth-first search tree from an arbitrary root node and testing whether all nodes are reachable from the root node. Tarjan's algorithm for strongly connected components \review{\cite{Tarjan72depthfirst}} can be used for directed graphs. 
Hence, both the {\em intersection} operation and the {\em connectivity testing} operation have similar costs. 
In what follows, we will refer to them as {\em elementary} operations. One advantage of using these elementary operations is that the high-level logic of the algorithms becomes elegant and simple. Also, their cost can be counterbalanced by the fact that they are highly generic and thus could benefit from dedicated circuits \review{(e.g., \joe{similar to} what has been done for all-pairs shortest-paths computations with FPGAs~\cite{FPGA})} or optimized code.
\end{observation}

\boldpara{Naive Upper Bound.}
\review{It is not hard to show} 
that both problems are solvable using $O(\delta^2)$ elementary operations based on a naive strategy. \review{\joe{A naive algorithm computes all the rows of the intersection hierarchy, using the fact that each graph $G_{(i,j)}$ can be obtained from the intersection of the two graphs immediately below it, i.e. $G_{(i,j)} = G_{(i,j-1)}\cap G_{(i+1,j)}$.}} For instance, $G_{(3,6)}=G_{(3,5)}\cap G_{(4,6)}$ in Figure~\ref{fig:example}. Hence, each row $k$ can be computed from row $k-1$ using $O(\delta)$ binary intersections. In the case of \tinter, \review{\joe{one must compute all rows}} until the $T^{th}$ row, and then answer {\tt true} iff all graphs in this row are connected. The total cost is $O(\delta T)=O(\delta^2)$ binary intersections, plus $\delta-T+1=O(\delta)$ connectivity tests for the $T^{th}$ row. Solving \inter is similar except that one needs to test the connectivity of all new graphs during the process. If a disconnected graph is first found in some row $k$, then the answer is $k-1$. If all graphs are connected up to row $\delta$, then $\delta$ is the answer. Since there are $O(\delta^2)$ graphs in the intersection hierarchy, the total number of connectivity tests and binary intersections is $O(\delta^2)$.

\bigskip
\boldpara{Lower Bound.}
The following lower bound is valid for any algorithm that uses only the two elementary operations {\em binary intersection} and {\em connectivity test}.

\begin{lemma}\label{lem:lower-bound}
  $\Omega(\delta)$ elementary operations are necessary to solve \tinter. 
\end{lemma}

\begin{proof}[by contradiction]
  Let ${\cal A}$ be an algorithm that uses only \review{intersection and connectivity testing} 
  and that decides whether any sequence of graphs is $T$-interval connected \review{using less than $\delta$ operations}. Then, for any sequence $\G$, at least one graph in $\G$ is never accessed by ${\cal A}$. Let $\G_1$ be a sequence that is $T$-interval connected and suppose that ${\cal A}$ decides that $\G_1$ is $T$-interval connected without accessing graph $G_k$. Now, consider a sequence $\G_2$ that is identical to $\G_1$ except $G_k$ is replaced by a disconnected graph $G_k'$.
Since $G_k'$ is never accessed, the executions of ${\cal A}$ on $\G_1$ and $\G_2$ are identical and ${\cal A}$ incorrectly decides that $\G_2$ is $T$-interval connected. \qed
\end{proof}

A similar argument can be used for \inter by making the answer $T$ dependent on the graph $G_k$ that is never accessed.

\section{Row-Based Strategy}
\label{sec:intermediate}

In this section, we present a basic strategy that improves upon the previous naive strategy, yielding upper bounds of $O(\delta \log \delta)$ operations for both problems. Its main interest is in the fact that it can be parallelized, and this allows us to show that both problems are in {\bf NC}, i.e. parallelizable on a PRAM with a polylogarithmic running time \review{\cite{GR88,Ja92}}. 
We first describe the algorithms for a sequential machine (RAM). \\

\review{
Informally, the general strategy is to compute only some of the rows of $\G$'s intersection hierarchy, \joe{using the fact that most graphs of the hierarchy can be obtained by intersecting two graphs that are several rows below (see Figure~\ref{fig:power}).
}}

\joe{The computation of each graph on the $\ell^{th}$ row involves the intersection of two graphs in the same row $\G^k$, and we need $k \geq \ell/2$ because the height of an intersection graph equals the length of the sequence that it represents in the graph $\G$.  Hence, to compute a row $\ell$, it is enough to compute only the rows below it with heights that are powers of 2, namely rows $\G^{2^{i}}$ which we call ``power rows'' (see Figure~\ref{fig:power}).}

\review{
More formally, the computation of ``power rows'' is based on the following lemma.}

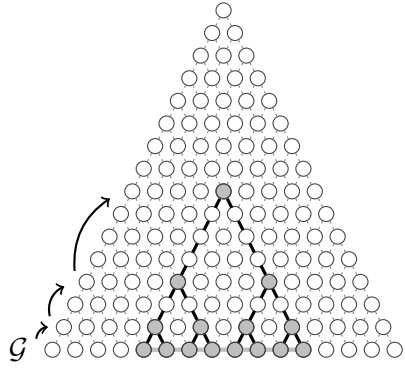
\begin{figure}
\begin{center}
\xdefinecolor{bordeaux}{rgb}{0.7,0.10,0.10}
\begin{tikzpicture}[scale=0.3]

\node (T) [font=\large] at (-0.5,-16) {\small $\G$};

\node (l1) [] at (0.5,-16) {};
\node (l2) [] at (1,-15) {};
\node (l4) []  at (2,-13) {};
\node (l8) [] at (4,-9) {};
\node (l10) [] at (5,-7.2) {};
\node (l11) []   at (5.5,-6) {};
\node (l12) []    at (6,-5) {};
\node (l16) []   at (8,-1) {};

\tikzstyle{searchPath}=[->,thick]
\tikzstyle{row}=[draw,rectangle,rounded corners]


\draw[searchPath] (l1) to[bend left=40] (0.8,-15);
\draw[searchPath] (l2) to[bend left=40] (l4);
\draw[searchPath] (l4) to[bend left] (l8);

  \tikzstyle{every node}=[draw, circle, inner sep=2pt,darkgray]
  \foreach \j in {16}{
     \foreach \i in {6,...,12}{
       \pgfmathsetmacro{\ii}{\i + (16-\j)/2};
       \draw[lightgray,ultra thick] (\ii-1,-\j) -- (\ii,-\j);
     }
  }

     \foreach \i in {5,...,12}{
       \pgfmathsetmacro{\ii}{\i};
       \path (\ii,-16) node[fill=lightgray] (\i16){};
	}

   \foreach \i in {5,7,9,11}{
       \pgfmathsetmacro{\ii}{\i + (16-15)/2};
       \path (\ii,-15) node[fill=lightgray] (\i15){};
	}

   \foreach \i in {5,9}{
       \pgfmathsetmacro{\ii}{\i + (16-13)/2};
       \path (\ii,-13) node[fill=lightgray] (\i13){};
	}

   \foreach \i in {5}{
       \pgfmathsetmacro{\ii}{\i + (16-9)/2};
       \path (\ii,-9) node[fill=lightgray] (\i9){};
	}

   
  \tikzstyle{every node}=[draw, circle, inner sep=2pt,darkgray]
  \foreach \j in {1,...,16}{
     \foreach \i in {1,...,\j}{
       \pgfmathsetmacro{\ii}{\i + (16-\j)/2};
       \path (\ii,-\j) node[thin] (\i\j){};
     }
  }

  \tikzstyle{every path}=[thin,dashed,gray,dash pattern=on 1pt off 1.5pt]
  
  \foreach \j in {1,...,15}{
     \foreach \i in {1,...,\j}{
       \pgfmathtruncatemacro{\jj}{\j + 1};
       \pgfmathtruncatemacro{\ii}{\i + 1};
       \draw (\i\jj) -- (\i\j);
       \draw (\ii\jj) -- (\i\j);
     }
  }
  
    \tikzstyle{every path}=[very thick,dashed,black,dash pattern=on 3.8pt off 5.6pt]

  \draw (516) -- (515);
  \draw (616) -- (515);
  \draw (716) -- (715);
  \draw (816) -- (715);
  \draw (916) -- (915);
  \draw (1016) -- (915);
  \draw (1116) -- (1115);
  \draw (1216) -- (1115);
  \draw (515) -- (513);
  \draw (715) -- (513);
  \draw (915) -- (913);
  \draw (1115) -- (913);
  \draw (513) -- (59);
  \draw (913) -- (59);

\end{tikzpicture}
\end{center}
\caption{\review{Example of computation of the intersection graph $\G^{8}[5]$ corresponding to the sequence $\{G_5, G_6, G_7, G_8, G_9, G_{10}, G_{11}, G_{12}\}$. \joe{The grey nodes in rows 2 and 4 and the initial sequence (in row 1) are the only graphs needed to compute $\G^{8}[5]$ in row 8.}}}\label{fig:power}
\end{figure} 

\begin{lemma}\label{lem:row-based}
  If some row $\G^{k}$ is already computed, then any row $\G^{\ell}$ for  $k+1 \leq \ell \leq 2k$ can be computed with $O(\delta)$ elementary operations.
\end{lemma}

\begin{proof}
  Assume that row $\G^{k}$ is already computed and that one wants to compute row $\G^{\ell}$ for some $k+1 \leq \ell \leq 2k$. Note that row $\G^{\ell}$ consists of the entries $\IG{\ell}{1},\ldots,\IG{\ell}{\delta-\ell+1}$.
Now, observe that for any $k+1 \leq \ell \leq 2k$ and for any $1 \leq i \leq \delta-\ell+1$, 
$\IG{\ell}{i} = G_{(i,i+\ell-1)} = G_{(i,i+k-1)} \cap G_{(i+\ell-k,i+\ell-1)} = \IG{k}{i} \cap \IG{k}{i+\ell-k}$.
Hence, $\delta-\ell+1 = O(\delta)$ intersections are sufficient to compute all of the entries of row $\G^{\ell}$. \qed
\end{proof}

\boldpara{\boldmath $T$-\inter.}
Using Lemma~\ref{lem:row-based}, we can incrementally compute ``power rows'' $\G^{2^{i}}$ for all $i$ from $1$ to $\lceil\log_2 T\rceil-1$ without computing the intermediate rows. Then, we compute row $\G^T$ directly from row $\G^{2^{\lceil\log_2 T\rceil-1}}$ (again using Lemma~\ref{lem:row-based}). This way, we compute $\lceil\log_2 T\rceil=O(\log \delta)$ rows, \review{\joe{using at most $\delta$ intersections per row}. Finally, we perform  $O(\delta)$ connectivity tests to test the connectivity of all graphs in $\G^T$. So, the total cost is $O(\delta \log \delta)$ elementary operations.}\\

\review{\joe{
In Figure~\ref{fig:power_tinter}, we show an example of an execution that tests $T$-interval connectivity based on power rows for a graph $\G$ of length $\delta=16$ with $T=11$. The computation of the power rows stops upon reaching $\G^{8}$ which is the last power row before $\G^{11}$. The graphs of the $T^{th}$ row are then computed from intersection graphs in $\G^{8}$ according to Lemma~\ref{lem:row-based}. For example, $\G^{11}[2]= \G^{8}[2] \cap \G^{8}[5]$. Finally, $T$-interval connectivity is verified by testing the connectivity of each graph in $\G^{11}$.\\
}}

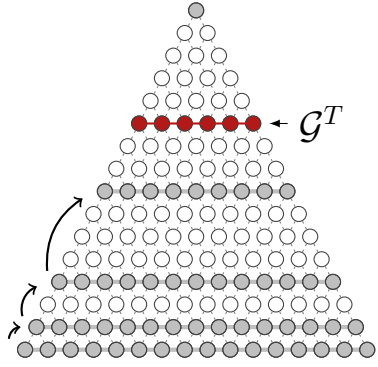
\begin{figure}
\begin{center}
\xdefinecolor{bordeaux}{rgb}{0.7,0.10,0.10}
\begin{tikzpicture}[scale=0.3]
\node (T) [font=\large] at (14,-6) {$\G^T$}; \draw[->, >=latex] (T) to (11.7,-6);

\node (l1) [] at (0.5,-16) {};
\node (l2) [] at (1,-15) {};
\node (l4) []  at (2,-13) {};
\node (l8) [] at (4,-9) {};
\node (l10) [] at (5,-7.2) {};
\node (l11) []   at (5.5,-6) {};
\node (l12) []    at (6,-5) {};
\node (l16) []   at (8,-1) {};

\tikzstyle{searchPath}=[->,thick]
\tikzstyle{row}=[draw,rectangle,rounded corners]


\draw[searchPath] (l1) to[bend left=40] (0.8,-15);
\draw[searchPath] (l2) to[bend left=40] (l4);
\draw[searchPath] (l4) to[bend left] (l8);

  \tikzstyle{every node}=[draw, circle, inner sep=2pt,darkgray]
  \foreach \j in {9,13,15,16}{
     \foreach \i in {2,...,\j}{
       \pgfmathsetmacro{\ii}{\i + (16-\j)/2};
       \draw[lightgray,ultra thick] (\ii-1,-\j) -- (\ii,-\j);
     }
  }
 
  \foreach \j in {1,9,13,15,16}{
     \foreach \i in {1,...,\j}{
       \pgfmathsetmacro{\ii}{\i + (16-\j)/2};
       \path (\ii,-\j) node[fill=lightgray] (\i\j){};
     }
  }
   
  \tikzstyle{every node}=[draw, circle, inner sep=2pt,darkgray]
  \foreach \j in {1,...,16}{
     \foreach \i in {1,...,\j}{
       \pgfmathsetmacro{\ii}{\i + (16-\j)/2};
       \path (\ii,-\j) node[thin] (\i\j){};
     }
  }

  \foreach \j in {6}{
     \foreach \i in {1,...,\j}{
       \pgfmathsetmacro{\ii}{\i + (16-\j)/2};
       \path (\ii,-\j) node[fill=bordeaux] (\i\j){};
     }
  }

  \foreach \j in {6}{
     \foreach \i in {2,...,\j}{
       \pgfmathsetmacro{\ii}{\i + (16-\j)/2};
       \draw[bordeaux,thick] (\ii-1,-\j) -- (\ii,-\j);
     }
  }

  \tikzstyle{every path}=[thin,dashed,gray,dash pattern=on 1pt off 1.5pt]
  \foreach \j in {1,...,15}{
     \foreach \i in {1,...,\j}{
       \pgfmathtruncatemacro{\jj}{\j + 1};
       \pgfmathtruncatemacro{\ii}{\i + 1};
       \draw (\i\jj) -- (\i\j);
       \draw (\ii\jj) -- (\i\j);
     }
  }
  
\end{tikzpicture}
\end{center}
\caption{\review{Example of $T$-interval connectivity testing based on the computation of power rows. Red (dark) nodes are the graphs of the $T^{th}$ row.}}\label{fig:power_tinter}

\end{figure}

\boldpara{\inter.}
Here, we incrementally compute rows $\G^{2^{i}}$ until we find a row that contains a disconnected graph (thus, a connectivity test is performed after each intersection). By Lemma~\ref{lem:row-based}, each of these rows can be computed using $O(\delta)$ intersections.
Suppose that row $\G^{2^{j+1}}$ is the first power row that contains a disconnected graph, and that $\G^{2^j}$ is the row computed before $\G^{2^{j+1}}$.  Next, we do
a binary search of the rows between $\G^{2^j}$ and $\G^{2^{j+1}}$ to find the row $\G^{T}$ with the highest row number $T$ such that all graphs on this row are connected (see Figure~\ref{power} for an illustration of the algorithm). The computation of each of these rows is based on row $\G^{2^j}$ and takes $O(\delta)$ intersections by Lemma~\ref{lem:row-based}. Overall, we compute at most $2 \lceil\log_2 T\rceil = O(\log \delta)$ rows using $O(\delta \log \delta)$ intersections and the same number of connectivity tests.

\review{\joe{Figure \ref{power} shows an example of interval connectivity testing based on the computation of power rows and a binary search between the last two power rows computed.} In this example $\delta=16$ and the found value of $T$ is $11$. The computation of power rows stops upon reaching $\G^{16}$, i.e. the first power row containing a disconnected graph ($\times$). So $T$ must be between $8$ (included) and $16$ (excluded). A binary search between rows $\G^8$ and $\G^{16}$ is then used to find $\G^{11}$, the highest row where all graphs are connected.\\
}

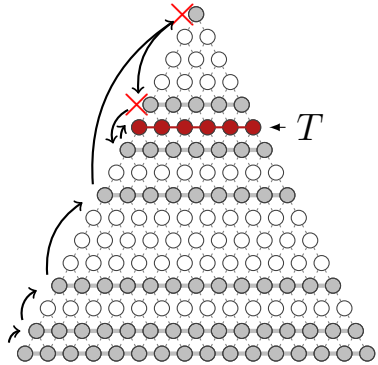
\begin{figure}

\begin{center}
\xdefinecolor{bordeaux}{rgb}{0.7,0.10,0.10}
\begin{tikzpicture}[scale=0.3]

\node (l1) [] at (0.5,-16) {};
\node (l2) [] at (1,-15) {};
\node (l4) []  at (2,-13) {};
\node (l8) [] at (4,-9) {};
\node (l10) [] at (5,-7.2) {};
\node (l11) []   at (5.5,-6) {};
\node (l12) []    at (6,-5) {};
\node (l16) []   at (8,-1) {};

\node () [font=\Large, red] at (7.9,-1) {$\times$};
\node () [font=\Large, red] at (5.9,-5) {$\times$};
\node (T) [font=\large] at (13.5,-6) {$T$};

\tikzstyle{searchPath}=[->,thick]
\tikzstyle{row}=[draw,rectangle,rounded corners]

\draw[->, >=latex] (T) to (11.7,-6);
\draw[searchPath] (l1) to[bend left=40] (0.8,-15);
\draw[searchPath] (l2) to[bend left=40] (l4);
\draw[searchPath] (l4) to[bend left] (l8);
\draw[searchPath] (l8) to[bend left] (l16);
\draw[searchPath] (l16) to[bend right] (l12);
\draw[searchPath] (l12) to[bend right=40] (l10);
\draw[searchPath] (5.2,-6.5) to[bend left=20] (5.4,-5.9);

  \tikzstyle{every node}=[draw, circle, inner sep=2pt,darkgray]
  \foreach \j in {5,7,9,13,15,16}{
     \foreach \i in {2,...,\j}{
       \pgfmathsetmacro{\ii}{\i + (16-\j)/2};
       \draw[lightgray,ultra thick] (\ii-1,-\j) -- (\ii,-\j);
     }
  }
  \foreach \j in {1,5,7,9,13,15,16}{
     \foreach \i in {1,...,\j}{
       \pgfmathsetmacro{\ii}{\i + (16-\j)/2};
       \path (\ii,-\j) node[fill=lightgray] (\i\j){};
     }
  }
  \tikzstyle{every node}=[draw, circle, inner sep=2pt,darkgray]
  \foreach \j in {1,...,16}{
     \foreach \i in {1,...,\j}{
       \pgfmathsetmacro{\ii}{\i + (16-\j)/2};
       \path (\ii,-\j) node[thin] (\i\j){};
     }
  }

  \foreach \j in {6}{
     \foreach \i in {1,...,\j}{
       \pgfmathsetmacro{\ii}{\i + (16-\j)/2};
       \path (\ii,-\j) node[fill=bordeaux] (\i\j){};
     }
  }
  \foreach \j in {6}{
     \foreach \i in {2,...,\j}{
       \pgfmathsetmacro{\ii}{\i + (16-\j)/2};
       \draw[bordeaux,thick] (\ii-1,-\j) -- (\ii,-\j);
     }
  }

  \tikzstyle{every path}=[thin,dashed,gray,dash pattern=on 1pt off 1.5pt]
  \foreach \j in {1,...,15}{
     \foreach \i in {1,...,\j}{
       \pgfmathtruncatemacro{\jj}{\j + 1};
       \pgfmathtruncatemacro{\ii}{\i + 1};
       \draw (\i\jj) -- (\i\j);
       \draw (\ii\jj) -- (\i\j);
     }
  }
\end{tikzpicture}
\end{center}
\caption{Example of interval connectivity testing based on the computation of power rows. Here $\delta=16$ and $T=11$.\label{power}}

\end{figure}

Now we establish that these problems are in {\bf NC} by showing that our algorithms are efficiently parallelizable. 

\begin{lemma}\label{lem:row-based-parallel}
  If some row $\G^{k}$ is already computed, then any row between $\G^{k+1}$ and $\G^{2k}$ can be computed in $O(1)$ time on an EREW PRAM with $O(\delta)$ processors.
\end{lemma}

\begin{proof}
  Assume that row $\G^{k}$ is already computed, and that one wants to compute row $\G^{\ell}$, consisting of the entries $\IG{\ell}{1},\ldots,\IG{\ell}{\delta-\ell+1}$, for some $k+1 \leq \ell \leq 2k$.
Since $\IG{\ell}{i} = \IG{k}{i} \cap \IG{k}{i+\ell-k}$, $1 \leq i \leq \delta-\ell+1$, the computation of row $\G^{\ell}$ can be implemented on an EREW PRAM with $\delta-\ell+1$ processors in two rounds as follows.
Let $P_i$, $1 \leq i \leq \delta-\ell+1$, be the processor dedicated to computing $\IG{\ell}{i}$. In the first round $P_i$ reads $\IG{k}{i}$, and in the second round $P_i$ reads $\IG{k}{i+\ell-k}$. This guarantees that each $P_i$ has exclusive access to the entries of row $\G^{k}$ that it needs for its computation.
Hence, row $\G^{\ell}$ can be computed in $O(1)$ time on an EREW PRAM using $O(\delta)$ processors. \qed
\end{proof}

\boldpara{\boldmath $T$-\inter on an EREW PRAM.}
The sequential algorithm for this problem computes $O(\log \delta)$ rows. By Lemma~\ref{lem:row-based-parallel}, each of these rows can be computed in $O(1)$ time on an EREW PRAM with $O(\delta)$ processors. Therefore, all of the rows (and hence all necessary intersections) can be computed in $O(\log \delta)$ time with $O(\delta)$ processors.
The $O(\delta)$ connectivity tests for row $\G^T$ can be done in $O(1)$ time with $O(\delta)$ processors. Then, the processors can establish whether or not all graphs in row $\G^T$ are connected by computing the logical AND of the results of the $O(\delta)$ connectivity tests in time $O(\log \delta)$ on a EREW PRAM with $O(\delta)$ processors using standard techniques (see~\cite{GR88,Ja92}). The total time is $O(\log \delta)$ on an EREW PRAM with $O(\delta)$ processors.\\

\boldpara{\inter on an EREW PRAM.}
The sequential algorithm for this problem computes $O(\log \delta)$ rows. 
Differently from \tinter, a connectivity test is done for each of the computed graphs (rather than just those of the last row) and it has to be determined for each computed row whether or not all of the graphs are connected. This takes $O(\log \delta)$ time for each of the $O(\log \delta)$ computed rows using the same techniques as for \tinter.
The total time is $O(\log^2 \delta)$ on an EREW PRAM with $O(\delta)$ processors.

\section{Optimal Solution}
\label{sec:optimal}

We now present our strategy for solving both \tinter and \inter using a linear number of elementary operations (in the length $\delta$ of $\G$),
matching the $\Omega(\delta)$ lower bound presented in Section~\ref{sec:definitions}.
The strategy relies on the concept of {\em ladder}. Informally, a ladder is a sequence of graphs that ``climbs'' the intersection hierarchy bottom-up.

\begin{definition}
The {\em right ladder of length $l$ at index $i$}, denoted by $\rlad{l}{i}$, is the sequence of intersection graphs $(\IG{k}{i}, \, k=1,2,\ldots,l)$. The {\em left ladder of length $l$ at index $i$}, denoted by $\llad{l}{i}$, is the sequence $(\IG{k}{i-k+1}, \, k=1,2,\ldots,l)$. A right ({\it resp.} left) ladder of length $l-1$ at index $i$ is said to be {\em incremented} when graph $\IG{l}{i}$ ({\it resp.} $\IG{l}{i-l+1}$) is added to it, and the resulting sequence of intersection graphs is called the {\em increment} of that ladder.
\end{definition}

\begin{lemma}
  \label{lem:ladder}
  A ladder of length $l$ can be computed using $l-1$ binary intersections.
\end{lemma}

\begin{proof}
  Consider a right ladder $\rlad{l}{i}$. For any $k \in [2,l]$ it holds that $\IG{k}{i}=\IG{k-1}{i}\cap G_{i+k-1}$. Indeed, by definition, $\IG{k-1}{i}=\cap \{G_i, G_{i+1},..., G_{i+k-2}\}$. The ladder can thus be built bottom-up using a single new intersection at each level \review{\joe{(in particular, the binary intersection of the previous graph in the ladder and one of the graphs from the sequence $\G$)}}.
 
  Consider a left ladder $\llad{l}{i}$. For any $k \in [2,l]$ it holds that $\IG{k}{i-k+1}=G_{i-k+1} \, \cap \, \IG{k-1}{i-k+2}$. Indeed, by definition, $\IG{k-1}{i-k+2}=\cap \{G_{i-k+2}, \allowbreak G_{i-k+3},..., G_{i}\}$. The ladder can thus be built bottom-up using a single new intersection at each level. \qed
\end{proof}

\begin{lemma}
  \label{lem:intersection}
Given $\llad{l_{\ell}}{\il}$ and $\rlad{l_r}{\ir}$, 
any pair \review{of indices} $(i,k)$ such that $\ir-l_{\ell} \le i<\ir$ and $\ir-i < k \le \ir-i+l_r$,
 $\IG{k}{i}$ can be computed by a single binary intersection, namely $\IG{k}{i} = \IG{\ir-i}{i} \cap \IG{k-\ir+i}{\ir}$. 
\end{lemma}

\begin{figure}[h]

\xdefinecolor{bordeaux}{rgb}{0.7,0.10,0.10}
\centering
\begin{tikzpicture}[scale=.36]
\begin{scope}
  
  \tikzstyle{every node}=[draw, circle, inner sep=2.5pt,lightgray!70]
  \foreach \j in {9,...,16}{
     \foreach \i in {1,...,\j}{
       \pgfmathsetmacro{\ii}{\i + (16-\j)/2};
       \path (\ii,-\j) node (\i\j){};
     }
  }

  \tikzstyle{every path}=[thin,dashed,dash pattern=on 1pt off 1.5pt,lightgray!70]
  \foreach \j in {9,...,15}{
     \foreach \i in {1,...,\j}{
       \pgfmathtruncatemacro{\jj}{\j + 1};
       \pgfmathtruncatemacro{\ii}{\i + 1};
       \draw (\i\jj) -- (\i\j);
       \draw (\ii\jj) -- (\i\j);
     }
  }

  \tikzstyle{every node}=[draw, circle, inner sep=2.2pt, fill=lightgray, lightgray]
  \path (815) node {};
  \path (814) node {};
  \path (714) node {};
  \path (713) node {};
  \path (613) node {};
  \path (612) node {};
  \path (512) node {};
  \path (511) node {};
  \path (411) node {};
  \path (410) node {};
  \path (310) node {};
  \path (310)+(.5,1) node {};

  \tikzstyle{every node}=[draw, circle, inner sep=2.2pt, fill=gray, gray]
  \draw[ultra thick,solid] (816) node {}--(715) node {}--(614) node {}--(513) node {}--(412) node {}--(311)  node {};
  \draw[ultra thick,solid] (916) node {}-- (915) node {};

  \tikzstyle{every node}=[draw, circle, inner sep=2.5pt, fill=bordeaux]
  \path (511) node {};
\end{scope}
\node () [] at (8.8,-11.0) {\large $\scriptstyle (i,k)$};
\node () [] at (2.2,-10.9) {\large $\scriptstyle l_\ell \rightarrow$};
\node () [] at (4.4,-13.1) {\large $\scriptstyle (i,\ir-i)$};
\node () [] at (12.4,-15.0) {\large $\scriptstyle (\ir,k-\ir+i)$};
\node () [] at (16.9,-15.0) {\large $\scriptstyle \leftarrow l_r$};
\node (next) [inner sep=1pt] at (8.3,-17) {$\ir$};
  \draw[->,>=latex] (next)  to (916);
\tikzstyle{inter}=[-,thick]
\draw[inter] (511) to[bend left=10] (513);
\draw[inter] (511) to[bend right=10] (915);
\end{tikzpicture}
\caption{\label{fig:ladders}Examples of intersections based on left and right ladders.}
\end{figure}
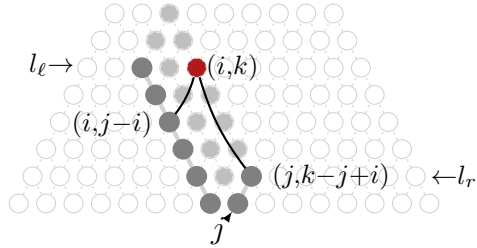

\review{Informally, the constraints $\ir-l_{\ell} \le i<\ir$ and $\ir-i < k \le \ir-i+l_r$ in Lemma~\ref{lem:intersection} define a rectangle delimited by two ladders and two lines, each of which is parallel to one of the ladders, as shown in Figure~\ref{fig:ladders}. The pairs $(i,k)$ defined by the constraints, shown in light grey in the figure, include all pairs that are strictly inside the rectangle, and all pairs on the parallel lines, but pairs on the two ladders are excluded.}

\begin{proof}[of Lemma~\ref{lem:intersection}]
By definition, $\IG{k}{i} = \cap \{G_i, G_{i+1},..., G_{i+k-1}\}$ and
$\IG{\ir-i}{i} = \cap \{G_i,\\ G_{i+1},$ $ ..., G_{\ir-1}\}$ and
$\IG{k-\ir+i}{\ir} = \cap \{G_{\ir}, G_{\ir+1},..., G_{i+k-1}\}$.
It follows that $\IG{k}{i} = \IG{\ir-i}{i} \cap \IG{k-\ir+i}{\ir}$. 
By definition, $\IG{\ir-i}{i} \in \llad{l_{\ell}}{\il}$ and $\IG{k-\ir+i}{\ir} \in \rlad{l_r}{\ir}$, so only a single binary intersection is needed.
\qed
\end{proof}

\bigskip
\boldpara{\boldmath $T$-\inter.}
We describe our optimal algorithm for this problem with reference to Figure~\ref{fig:optimal-T} which shows two examples of the execution of the algorithm (see Algorithm~\ref{algo:optimal-T} for details).
The algorithm traverses the $T^{th}$ row in the intersection hierarchy from left to right, starting at $\IG{T}{1}$. If a disconnected graph is found, the algorithm returns {\tt false} and terminates. If the algorithm reaches the last graph in the row, i.e. $\IG{T}{\delta-T+1}$, and no disconnected graph was found, then it returns {\tt true}.
The graphs $\IG{T}{1}, \IG{T}{2}, \ldots, \IG{T}{\delta-T+1}$ are computed based on the set of ladders
$\mathcal{S} = \{\llad{T}{T},$ $ \rlad{T-1}{T+1}, \llad{T}{2T}$, $\rlad{T-1}{2T+1}, \ldots\}$, which are constructed as follows. 
Each left ladder is built entirely (from bottom to top) 
when the traversal arrives at its top location in row $T$ (i.e. where the last increment is to take place). For instance, $\llad{T}{T}$ is built when the walk is at index $1$ in row $T$, $\llad{T}{2T}$ is built at index $T+1$, and so on.  If a disconnected graph is found in the process, the execution terminates returning {\tt false}.

\noindent

\begin{figure}
\centering
\begin{tikzpicture}[scale=0.33]
\begin{scope}
  \node (T) [font=\Large] at (1,-12) {$\G^T$};

   \tikzstyle{every node}=[draw, inner sep=4pt, thin, darkgray]
  \path (3,-12) node[circle] {};
   \path (8,-12) node[circle] {};
    \path (13,-12) node[circle] {};

  \tikzstyle{every node}=[draw, circle, inner sep=2pt,darkgray]
  \foreach \j in {9,...,16}{
     \foreach \i in {1,...,\j}{
       \pgfmathsetmacro{\ii}{\i + (16-\j)/2};
       \path (\ii,-\j) node (\i\j){};
     }
  }

\draw[->] (T) to (2.5,-12);

  \tikzstyle{inter}=[-]

\draw[inter] (212) to[bend left=10] (213);
\draw[inter] (212) to[bend right=10] (616);

\draw[inter] (312) to[bend left=10] (314);
\draw[inter] (312) to[bend right=10] (615);

\draw[inter] (412) to[bend left=10] (415);
\draw[inter] (412) to[bend right=10] (614);

\draw[inter] (512) to[bend left=10] (516);
\draw[inter] (512) to[bend right=10] (613);

\draw[inter] (712) to[bend left=10] (713);
\draw[inter] (712) to[bend right=10] (1116);

\draw[inter] (812) to[bend left=10] (814);
\draw[inter] (812) to[bend right=10] (1115);

\draw[inter] (912) to[bend left=10] (915);
\draw[inter] (912) to[bend right=10] (1114);

\draw[inter] (1012) to[bend left=10] (1016);
\draw[inter] (1012) to[bend right=10] (1113);

\draw[inter] (1212) to[bend left=10] (1213);
\draw[inter] (1212) to[bend right=10] (1616);

      \tikzstyle{every path}=[thin,dashed,gray,dash pattern=on 1pt off 1.5pt]
  \foreach \j in {9,...,15}{
     \foreach \i in {1,...,\j}{
       \pgfmathtruncatemacro{\jj}{\j + 1};
       \pgfmathtruncatemacro{\ii}{\i + 1};
       \draw (\i\jj) -- (\i\j);
       \draw (\ii\jj) -- (\i\j);
     }
  }
  \foreach \i in {1,...,12}{
    \pgfmathsetmacro{\ii}{\i + 2};
    \path (\ii,-12) node {};
  }

\tikzstyle{every node}=[draw, circle, inner sep=2pt, fill=gray]
  \draw[ultra thick, solid] (616) node {}--(615) node {}--(614) node {}--(613) node {};
  \draw[ultra thick, solid] (1116) node {}-- (1115) node {}-- (1114) node {}-- (1113) node {};
    \draw[ultra thick, solid] (516) node {}--(415) node {}--(314) node {}--(213) node {}--(112);
    \draw[ultra thick, solid] (1016) node {}--(915) node {}--(814) node {}--(713) node {}--(612);
    \draw[ultra thick, solid] (1516) node {}--(1415) node {}--(1314) node {}--(1213) node {}--(1112) ;
        \draw[ultra thick, solid] (1616) node {};

\xdefinecolor{bordeaux}{rgb}{0.7,0.10,0.10}

  \tikzstyle{every node}=[draw, circle, inner sep=1.9pt, fill=bordeaux]
  \path (112) node {};
  \path (212) node {};
  \path (312) node {};
  \path (412) node {};
  \path (512) node {};
  \path (612) node {};
  \path (712) node {};
  \path (812) node {};
  \path (912) node {};
  \path (1012) node {};
  \path (1112) node {};
  \path (1212) node {};

\end{scope}
\begin{scope}[xshift=17cm]

\xdefinecolor{bordeaux}{rgb}{0.7,0.10,0.10}

\node (T) [font=\Large] at (2.5,-8) {$\G^T$};

 \tikzstyle{every node}=[draw, inner sep=4pt, thin, darkgray]
  \path (5,-8) node[circle] {};

  \tikzstyle{every node}=[draw, circle, inner sep=2pt,darkgray]
  \foreach \j in {7,...,16}{
     \foreach \i in {1,...,\j}{
       \pgfmathsetmacro{\ii}{\i + (16-\j)/2};
       \path (\ii,-\j) node (\i\j){};
     }
  }

  \draw[->] (T) to (4.5,-8);

 \tikzstyle{inter}=[-]

\draw[inter] (28) to[bend left=10] (29);
\draw[inter] (28) to[bend right=10] (1016);

\draw[inter] (38) to[bend left=10] (310);
\draw[inter] (38) to[bend right=10] (1015);

\draw[inter] (48) to[bend left=10] (411);
\draw[inter] (48) to[bend right=10] (1014);

\draw[inter] (58) to[bend left=10] (512);
\draw[inter] (58) to[bend right=10] (1013);

\draw[inter] (68) to[bend left=10] (613);
\draw[inter] (68) to[bend right=10] (1012);

\draw[inter] (78) to[bend left=10] (714);
\draw[inter] (78) to[bend right=10] (1011);

\draw[inter] (88) to[bend left=10] (815);
\draw[inter] (88) to[bend right=10] (1010);

  \tikzstyle{every path}=[thin,dashed,gray,dash pattern=on 1pt off 1.5pt]
  \foreach \j in {7,...,15}{
     \foreach \i in {1,...,\j}{
       \pgfmathtruncatemacro{\jj}{\j + 1};
       \pgfmathtruncatemacro{\ii}{\i + 1};
       \draw (\i\jj) -- (\i\j);
       \draw (\ii\jj) -- (\i\j);
     }
  }

  \foreach \i in {1,...,7}{
    \pgfmathsetmacro{\ii}{\i + 4.5};
    \path (\ii,-7) node {};
  }

\tikzstyle{every node}=[draw, circle, inner sep=2pt, fill=gray]
  \draw[ultra thick, solid] (1016) node {}--(1015) node {}--(1014) node {}--(1013) node {}--(1012) node {}--(1011) node {}--(1010) node {};
  \draw[ultra thick, solid] (916) node {}-- (815) node {}-- (714) node {}-- (613) node {}-- (512) node {}-- (411) node {}-- (310) node {}-- (29) node {}-- (18);

  \tikzstyle{every node}=[draw, circle, inner sep=1.9pt, fill=bordeaux]
  \path (18) node {};
  \path (28) node {};
  \path (38) node {};
  \path (48) node {};
  \path (58) node {};
  \path (68) node {};
  \path (78) node {};
  \path (88) node {};

\end{scope}

\end{tikzpicture}
\caption{\label{fig:optimal-T}Examples of the execution of
the optimal algorithm for \tinter with $T < \delta/2$ (left) and $T \geq \delta/2$ (right). $\G$ is $T$-interval connected in both examples.}
\end{figure}
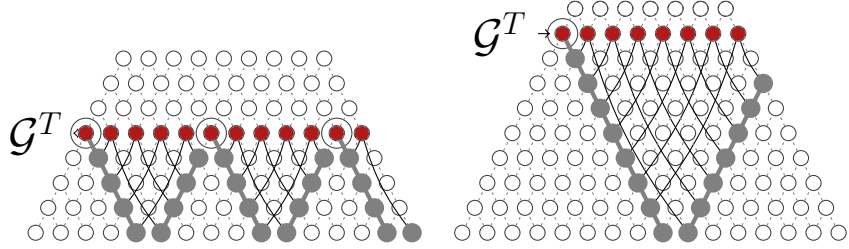

\begin{algorithm}[htb]
\footnotesize
$k \leftarrow T$ \hfill // current row (non-changing) \\
$i \leftarrow 1$ \hfill // current index in the row \\
$next \leftarrow 1$ \hfill // trigger for next ladder construction \\
\BlankLine
// walk until stepping out of the intersection hierarchy \\
\While{$i \leq \delta-k+1$}{
\eIf{$i=next$}{$next\gets i+k$\\ \If{$\neg$\texttt{computeFromRight($k,i,next$)}}{\tt return false}}{
\texttt{computeFromIntersection($k,i,next$)}\\  
}
\If{$\neg $\texttt{isConnected($\IG{k}{i}$)}}{\texttt{return false}}
$i \leftarrow i+1$\\
}
\texttt{return true}
\BlankLine
\hrule\medskip
  \SetKwBlock{Begin}{}{}

  function \texttt{computeFromRight($k,i,next$):} \hfill // compute the left ladder ${\cal L}^k[i]$
  \Begin{
    $k'\gets 1$ \hfill // row of first increment\\
    $i'\gets next-1$ \hfill // index of first increment\\
    \While{$k'<k$}{
      \If{$\neg${\tt isConnected}($\IG{k'}{i'}$)}{
        {\tt return false} \hfill // a disconnected graph was found
      }
      $k' \gets k' + 1$\\
      $i' \gets i' - 1$\\ 
     $\G^{k'}[i'] \gets \G^{k'-1}[i'+1] \cap G_{i'}$ \hfill// ``increment'' the ladder\\
    }
  }
  \BlankLine
  \hrule\medskip
  function \texttt{computeFromIntersection($k,i,next$):} \hfill // ``increment'' the right ladder
  \Begin{
    $k' \gets k-next+i$ \hfill // row of increment (right ladder) \\
    $\IG{k'}{next} \gets \IG{k'-1}{next} \cap G_{next+k'-1}$ \hfill // ``increment'' right ladder\\
    $\IG{k}{i} \gets \IG{next-i}{i} \cap \IG{k'}{next}$ \hfill // compute intersection based on Lemma~\ref{lem:intersection}\\
  }

 \caption{\label{algo:optimal-T} Optimal algorithm for \tinter}
\label{algo1}
\end{algorithm}

Differently from left ladders, right ladders are constructed gradually as the traversal proceeds. Each time that the traversal moves right to a new index in the $T^{th}$ row, the current right ladder is incremented and the new top element of this right ladder is used immediately to compute the graph at the current index in the $T^{th}$ row (using Lemma~\ref{lem:intersection}). This continues until the right ladder reaches row $T-1$ after which a new left ladder is built.

\joe{The following theorem establishes the complexity of our algorithm.}

\begin{theorem}\label{thm:T-inter-opt}
\tinter can be solved with $\Theta(\delta)$ elementary operations, which is optimal (to within a constant factor).
\end{theorem}
\begin{proof}
\review{
The set $\mathcal{S}$ of ladders constructed by the algorithm
includes at most $\lceil\delta/T\rceil$ left ladders and $\lceil\delta/T\rceil$ right ladders, each of length at most $T$. By Lemma~\ref{lem:ladder}, the set of ladders $\mathcal{S}$ can be computed using less than $2\delta$ binary intersections. Based on Lemma~\ref{lem:intersection}, each of the $\delta-T+1$ graphs $\IG{T}{i}$ in row $T$ can be computed at the cost of a single intersection of two graphs in $\mathcal{S}$. At most $\delta-T+1$ connectivity tests are performed for row $T$, which concludes the proof.
\qed
}
\end{proof}

\boldpara{\inter.}
The strategy of our optimal algorithm for this problem is in the same spirit as the one for \tinter. However, it is more complex and corresponds to a walk in the two dimensions of the intersection hierarchy. It is best understood with reference to Figure~\ref{fig:walk} which shows an example of the execution of the algorithm (see Algorithm~\ref{algo:optimal_interval} for details).

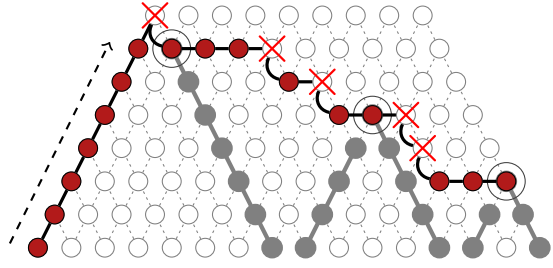
\begin{figure}
  \centering
\xdefinecolor{bordeaux}{rgb}{0.7,0.10,0.10}
\begin{tikzpicture}[scale=.44]
\tikzstyle{direction}=[->,thick,dashed]
\node (l1) [] at (0,-16.2) {};
\node (l8) [] at (4,-9) {};
\draw[direction] (l1) to[] (3.15,-9.8);

  \tikzstyle{every node}=[draw, circle, inner sep=2.5pt,gray]
  \foreach \j in {9,...,16}{
     \foreach \i in {1,...,\j}{
       \pgfmathsetmacro{\ii}{\i + (16-\j)/2};
       \path (\ii,-\j) node (\i\j){};
     }
  }
  \tikzstyle{every path}=[thin,dashed,gray,dash pattern=on 1pt off 1.5pt]
  \foreach \j in {9,...,15}{
     \foreach \i in {1,...,\j}{
       \pgfmathtruncatemacro{\jj}{\j + 1};
       \pgfmathtruncatemacro{\ii}{\i + 1};
       \draw (\i\jj) -- (\i\j);
       \draw (\ii\jj) -- (\i\j);
     }
  }
  \tikzstyle{every node}=[font=\LARGE, red]
  \path (19) node {$\times$};
  \path (510) node {$\times$};
  \path (711) node {$\times$};
  \path (1012) node {$\times$};
  \path (1113) node {$\times$};
  \tikzstyle{every node}=[draw, circle, inner sep=2.3pt, fill=gray]
  \draw[ultra thick, solid] (816) node {}--(715) node {}--(614) node {}--(513) node {}--(412) node {}--(311) node {}--(210) node {};
  \draw[ultra thick, solid] (916) node {}-- (915) node {}-- (914) node {}-- (913) node {};
  \draw[ultra thick, solid] (1316) node {}--(1215) node {}-- (1114) node {}--(1013) node {}-- (912) node {};
  \draw[ultra thick, solid] (1416) node {}-- (1415) node {};
  \draw[ultra thick, solid] (1414) node {}--(1515) node {}-- (1616) node {};
  \path (1616) node {};

  \tikzstyle{every path}=[]
  \tikzstyle{every node}=[draw, inner sep=5pt, thin, darkgray]
  \path (210) node[circle] {};
  \path (912) node[circle] {};
  \path (1414) node[circle] {};
  \tikzstyle{every node}=[draw, circle, inner sep=2.5pt, fill=bordeaux]
  \path (116) node {};
  \path (115) node {};
  \path (114) node {};
  \path (113) node {};
  \path (112) node {};
  \path (111) node {};
  \path (110) node {};
  \path (210) node {};
  \path (310) node {};
  \path (410) node {};
  \path (611) node {};
  \path (812) node {};
  \path (912) node {};
  \path (1214) node {};
  \path (1314) node {};
  \path (1414) node {};

  \tikzstyle{every path}=[very thick,rounded corners=6pt]
  \draw (116)--(115)--(114)--(113)--(112)--(111)--(110)--(19);
  \draw (19)--(110.north east)--(210)--(310);
  \draw (310)--(410)--(510)--(511.north east)--(611);
  \draw (611)--(711)--(712.north east)--(812);
  \draw (812)--(912)--(1012)--(1013.north east)--(1113);
  \draw (1113)--(1114.north east)--(1214);
  \draw (1214)--(1314)--(1414);

\end{tikzpicture}
\caption{\label{fig:walk}Example of the execution of
the optimal algorithm for \inter. {\it (It is a coincidence that the rightmost ladder matches the outer face.)}}
\end{figure}

\begin{algorithm}
\footnotesize

$k \leftarrow 1$ \hfill // current row\\
$i \leftarrow 1$ \hfill // current index in the row\\
$next \leftarrow 2$ \hfill // trigger for next ladder construction\medskip\\
// builds a right ladder until a disconnected graph is found\\
\While{ \texttt{isConnected($\IG{k}{1}$)} }{
  $k\gets k+1$\\
  \eIf{$k > \delta$}{
    {\tt return} $\delta$ \hfill // the graph is $\delta$-interval connected
  }{
    $\IG{k}{1} \leftarrow \IG{k-1}{1} \cap G_{k}$ \hfill // ``increment'' the right ladder
  }
}
\If{$k=1$}{{\tt return} $0$ \hfill // the graph is $0$-interval connected}\medskip
$k\gets k-1$ \hfill // move down\\
$i\gets i+1$ \hfill // move right\\

\BlankLine
\BlankLine

// walk until stepping out of the hierarchy\\
\While{$i \leq \delta-k+1$}{
  \eIf{$i=next$}{
    $next \gets i+k$\\
    {\tt computeFromRight($k,i,next$)}\\
  }{
    {\tt computeFromIntersection($k,i,next$)}\\
    \If {$\neg$\texttt{isConnected($\IG{k}{i}$)}}{
      $k\gets k-1$\\
    }
  }
  \If {$k=0$}{
    {\tt return 0}
  }
    $i\gets i+1$\\
  }
  \texttt{return} $k$
  \hrule\medskip
  \SetKwBlock{Begin}{}{}

  function \texttt{computeFromRight($k,i,next$):} \hfill // compute the left ladder ${\cal L}^k[i]$
  \Begin{
    $k'\gets 1$ \hfill // row of first increment\\
    $i'\gets next-1$ \hfill // index of first increment\\
    \While{$k'<k$}{
      \If{$\neg${\tt isConnected}($\IG{k'}{i'}$)}{
        $k \gets k'-1$ \hfill // move the original walk..\\
        $i \gets i'+1$ \hfill // ..below-right disconnected graph,\\
        {\tt return} \hfill // abort function
      }
      $k' \gets k' + 1$\\
      $i' \gets i' - 1$\\
      $\G^{k'}[i'] \gets \G^{k'-1}[i'+1] \cap G_{i'}$ \hfill // ``increment'' the ladder
    }
  }
  \hrule\medskip
  \SetKwBlock{Begin}{}{}
  function \texttt{computeFromIntersection($k,i,next$):} \hfill (Same function as for Algorithm~\ref{algo:optimal-T})
  \Begin{
    $k' \gets k-next+i$ \hfill // row of increment (right ladder)\\
    $\IG{k'}{next} \gets \IG{k'-1}{next} \cap G_{next+k'-1}$ \hfill // ``increment'' right ladder\\
    $\IG{k}{i} \gets \IG{next-i}{i} \cap \IG{k'}{next}$ \hfill // compute intersection based on Lemma~\ref{lem:intersection}
   }

 \caption{\label{algo:optimal_interval} Optimal algorithm for \inter}
\end{algorithm}

The walk starts at the bottom left graph $\IG{1}{1}$ and builds a right ladder incrementally until it encounters a disconnected graph. If $\IG{\delta}{1}$ is reached and is connected, then $\G$ is $\delta$-interval connected and execution terminates returning $\delta$. Otherwise, suppose that a disconnected graph is first found in row $k+1$. Then $k$ is an upper bound on the connectivity of $\G$ and the walk drops down a level to $\IG{k}{2}$ which is the next graph in row $k$ that needs to be checked. This requires the construction of a left ladder $\llad{k}{k+1}$ of length $k$ ending at $\IG{k}{2}$. The walk proceeds rightward on row $k$ using a similar traversal strategy to 
the algorithm for \tinter. Here, however, every time that a disconnected graph is found, the walk drops down by one row. The dropping down operation, say, from some $\IG{k}{i}$, is made in two steps (curved line in Figure~\ref{fig:walk}). First it goes to $\IG{k-1}{i}$, which is necessarily connected because $\IG{k}{i-1}$ is connected (so a connectivity test is not needed here), and then it moves one unit right to $\IG{k-1}{i+1}$. If the walk eventually reaches the rightmost graph of some row and this graph is connected, then the algorithm terminates returning the corresponding row number as $T$. Otherwise the walk will terminate at a disconnected graph in row 1 and $\G$ is not $T$-interval connected for any $T$. In this case, the algorithm returns $T=0$.

Similarly to the algorithm for \tinter,
 the computations of the graphs in a walk by Algorithm~\ref{algo:optimal_interval} (for \inter)
use binary intersections based on Lemmas~\ref{lem:ladder} and~\ref{lem:intersection}. If the algorithm returns that $\G$ is $T$-interval connected, then each graph $\IG{T}{1}, \IG{T}{2}, \ldots, \IG{T}{\delta-T+1}$ must be connected. The graphs that are on the walk are checked directly by the algorithm. For each graph $\IG{T}{i}$ on row $T$ that is below the walk, there is a graph
$\IG{j}{i}$ with $j>T$
that is on the walk and is connected and this implies that $\IG{T}{i}$ is connected.

\begin{theorem}\label{thm:inter-opt}
\inter can be solved with $\Theta(\delta)$ elementary operations, which is optimal (up to a constant factor).
\end{theorem}
\begin{proof}
\review{
The ranges of the indices covered by the left ladders that are constructed by the process are disjoint, so their total length is $O(\delta$). The first right ladder has length at most $\delta$ and each subsequent right ladder has length less than the left ladder that precedes it so the total length of the right ladders is also $O(\delta)$. Therefore, 
this algorithm  performs $O(\delta)$ binary intersections and $O(\delta$) connectivity tests. This establishes that {\sc \inter} can be solved with $\Theta(\delta)$ elementary operations matching the lower bound of Lemma~\ref{lem:lower-bound}. \hfill \qed
}
\end{proof}

\FloatBarrier
\boldpara{Online Algorithms.}
The optimal algorithms for \tinter and\linebreak
\inter can be adapted to an online setting in which the sequence of graphs $G_1,G_2,G_3,\ldots$ of a dynamic graph $\G$ is processed in the order that the graphs are received.  In the case of \tinter, the algorithm cannot provide an answer until at least $T$ graphs have been received. When the $T^{th}$ graph is received, the algorithm builds the first left ladder using $T-1$ binary intersections. It can then perform a connectivity test and answer whether or not the sequence is $T$-interval connected so far. After this initial period, a $T$-connectivity test can be performed for the $T$ most recently received graphs (by performing a connectivity test on the corresponding graph in row $T$) after the receipt of each new graph.

\begin{theorem}\label{thm:online}
\tinter and \inter can be solved online with an amortized cost of $\Theta(1)$ elementary operations per graph received.
\end{theorem}
\begin{proof}
\review{
At no time \joe{during the execution of the algorithm for \tinter}  does the number of intersections performed to build left ladders exceed the number of graphs received and the same is true for right ladders. Furthermore, each new graph after the first $T-1$ graphs corresponds to a graph in row $T$ which can be computed with one intersection by Lemma~\ref{lem:intersection}.  In summary, the amortized cost is $O(1)$ elementary operations for each graph received and for each $T$-connectivity test after the initial period. The analysis for {\sc \inter} is similar except that the algorithm can report the connectedness of the sequence received so far, starting with the first graph.
 \qed
}
\end{proof}
\section{Dynamic Online Interval Connectivity}
\label{sec:dynamic}

The algorithms in this section are motivated by Internet protocols like TCP (Transmission Control Protocol) which
adjust their behaviour dynamically in response to recent network events and conditions such as dropped packets and congestion. $T$-interval connectivity is a measure of the stability of a network. Generally, larger values of $T$ indicate that communication is more reliable, so it is natural to consider a dynamic version of interval connectivity that is based only on the recent states of a network rather than the entire history of a network.
We formalize this notion of recent history by introducing the concept of {\em $T$-stable} graphs. We then define the dynamic online versions of both \tinter and \inter in terms of $T$-stable graphs.

\begin{definition}[\boldmath $T$-stable graph]
A graph $G_i$, $i\geq T$, of a sequence $\G=(G_1,G_2,...,G_\delta)$ is {\em $T$-stable} for a given $T$ iff the subsequence $ G_{i-(T-1)}, G_{i-(T-2)}, \ldots , G_{i-1}, G_i$ is $T$-interval connected.
\end{definition}

\begin{definition}[Testing \boldmath \tstable]
The \tstable problem for a given $T$ is the problem of deciding for each received graph $G_i$, $i \geq T$, whether $G_i$ is $T$-stable.
\end{definition}

\begin{definition}[\boldmath \stable]
We use the term \stable to refer to the problem of finding $T_i =\max\{T :$ $G_i$ is $T$-stable$\}$ for each received graph $G_i$.
\end{definition}

As before, the first problem is a decision problem with true/false output, while the second is a maximization problem with integer output. Here, however, one such output is required after each graph in the sequence is received.

\medskip
\boldpara{\boldmath \tstable.} Our algorithm for \tstable is similar to Algorithm~\ref{algo:optimal-T} for \tinter. The differences are that the algorithm for \tstable produces an output after each graph of a sequence is received, and the algorithm does not terminate if a disconnected graph is found on row $T$ of the hierarchy. Instead, it continues until the last graph in the sequence is received. The ladders constructed by the algorithm for \tstable are the same as the ladders that would be constructed by Algorithm~\ref{algo:optimal-T} for a dynamic graph that is $T$-interval connected (see Figure~\ref{fig:optimal-T} for examples). Given a dynamic graph $\G=(G_1,G_2,...,G_\delta)$, \tstable is undefined for the graphs $G_i$ with $i<T$, so the algorithm returns $\bot$ after each of the first $T-1$ graphs is received. When $G_T$ is received, the algorithm builds a left ladder and returns {\tt true} (resp. {\tt false}) if the top graph of the ladder (i.e. $\G^T[1]$) is connected (resp. disconnected). Then the walk progresses rightward along row $T$ every time that a graph is received, alternately building left and right ladders in such a way that the graph $\G^T[i-(T-1)]$ can always be computed from $G_i$ with a single intersection (using Lemma~\ref{lem:intersection}). $G_i$ is $T$-stable iff $\G^T[i-(T-1)]$ is connected and {\tt true} or {\tt false} is output as appropriate.

\begin{theorem}\label{thm:tstable}
\tstable can be solved online with an amortized cost of $\Theta(1)$ elementary operations per graph received.
\end{theorem}
\begin{proof}
\review{
By the same analysis as for the online version of Algorithm~\ref{algo:optimal-T}, the number of intersections performed to build left ladders never exceeds the number of graphs received, and the same is true for the number of intersections to build right ladders and for the number of connectivity tests. \qed
}
\end{proof}

\medskip
\boldpara{\boldmath \stable.} The algorithm for this problem must find $T_i=\max\{T :$ $G_i$ is $T$-stable$\}$ for each received graph $G_i$.
Our algorithm for \stable generalizes the strategy that we used in the algorithm for \inter by trying to climb as high as possible in the hierarchy, even after a disconnected intersection graph is found. This is necessary because the sequence of values $T_1,T_{2},T_{3},\ldots$ for \stable is not necessarily monotonic.

The algorithm for \stable uses right and left ladders
to walk through the intersection hierarchy. The general idea is that the walk goes up when the current intersection graph is connected and down when it is disconnected (unless the walk is on the bottom level of the hierarchy in which case it goes right to the next graph). This is different from the algorithm for \inter which only goes up during the construction of the first right ladder and goes right or down in all other cases. We will describe the algorithm for \stable with reference to Figure~\ref{fig:window_walk} which shows an example of the execution of the algorithm. See Algorithm~\ref{algo:optimal_interval_window} for complete details.

\begin{figure}
  \centering
\xdefinecolor{bordeaux}{rgb}{0.7,0.10,0.10}
\begin{tikzpicture}[scale=.44]
\draw[->,thick] (14.7,-16.9) to[] (14.2,-16.4);
\draw[->,thick] (14.5,-11) to[] (14.0,-11);
  \node (T) [font=\footnotesize] at (17,-17.2) {$current$ $graph$ $G_{14}$};
    \node (T) [font=\small] at (15.2,-11) {$T_{14}$};

\tikzstyle{direction}=[->,thick,dashed]
\node (l1) [] at (0,-16.2) {};
\node (l8) [] at (4,-9) {};
\draw[direction] (l1) to[] (3.15,-9.8);

  \tikzstyle{every node}=[draw, circle, inner sep=2.5pt,gray]
  \foreach \j in {9,...,16}{
     \foreach \i in {1,...,\j}{
       \pgfmathsetmacro{\ii}{\i + (16-\j)/2};
       \path (\ii,-\j) node (\i\j){};
     }
  }
  \tikzstyle{every path}=[thin,dashed,gray,dash pattern=on 1pt off 1.5pt]
  \foreach \j in {9,...,15}{
     \foreach \i in {1,...,\j}{
       \pgfmathtruncatemacro{\jj}{\j + 1};
       \pgfmathtruncatemacro{\ii}{\i + 1};
       \draw (\i\jj) -- (\i\j);
       \draw (\ii\jj) -- (\i\j);
     }
  }

  \tikzstyle{every node}=[draw, circle, inner sep=2.3pt, fill=gray]
  \draw[ultra thick, solid] (616) node {}--(515) node {}--(414) node {}--(313) node {}--(212) node {};
  \draw[ultra thick, solid] (716) node {}-- (715) node {}-- (714) node {}-- (713) node {}-- (712) node {};
    \draw[ultra thick, solid] (1116) node {}--(1015) node {}--(914) node {} --(813) node {};
    \draw[ultra thick, solid] (1216) node {}--(1215) node {}--(1214) node {};

  \tikzstyle{every node}=[font=\LARGE, red]
  \path (29) node {$\times$};
  \path (111) node {$\times$};
   \path (310) node {$\times$};
   \path (411) node {$\times$};
  \path (712) node {$\times$};
    \path (813) node {$\times$};

  \tikzstyle{every path}=[]
  \tikzstyle{every node}=[draw, inner sep=5pt, thin, darkgray]
  \path (212) node[circle] {};
    \path (813) node[circle] {};

  \tikzstyle{every node}=[draw, circle, inner sep=2.5pt, fill=bordeaux]
  \path (116) node {};
  \path (115) node {};
  \path (114) node {};
  \path (113) node {};
  \path (112) node {};
  \path (212) node {};
  \path (211) node {};
  \path (210) node {};
  \path (512) node {};
  \path (511) node {};
  \path (914) node {};
  \path (913) node {};
  \path (912) node {};
  \path (911) node {};

  \tikzstyle{every path}=[very thick,rounded corners=6pt]
  \draw (116)--(115)--(114)--(113)--(112)--(111)--(112.north east)--(212)--(211)--(210)--(29)--(210.north east)--(310)--(311.north east)--(411)--(412.north east)--(512)--(511);
  \draw (712)--(713.north east)--(813)--(814.north east)--(914)--(913)--(912)--(911);

\end{tikzpicture}
\caption{\label{fig:window_walk}Example of the execution of the \stable algorithm.}
\end{figure}
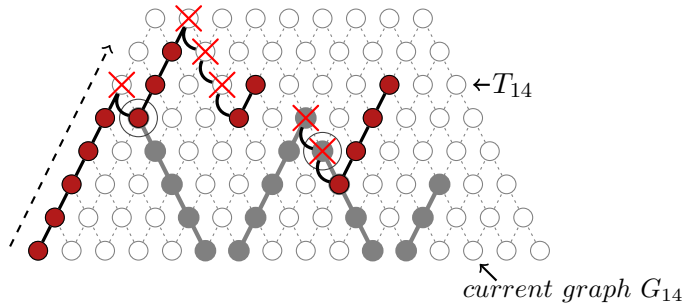

\begin{algorithm}
\footnotesize
$i \leftarrow 1$ \hfill // current index in the row\\
$next \leftarrow 2$ \hfill // trigger for next ladder construction\\
$output \leftarrow 0$ \medskip\\
\BlankLine
\BlankLine

\While{receiving graphs}{

\While{ \texttt{isConnected($\IG{k}{i}$)} }{ 
$output \gets k$; 
$k\gets k+1$\\
    {\tt computeFromIntersection($k,i,next$)} \hfill // ``increment'' the right ladder\\
}

\While{ \texttt{$\neg$isConnected($\IG{k}{i}$)} }{  
\eIf{$k=1$}{$output \gets 0$; $next \gets i+2$}{$k \gets k-1$\\} 
      $i\gets i+1$\\
  \eIf{$i=next$}{
    $next \gets i+k$; 
    {\tt computeFromRight($k,i,next$)}\\
  }{
    {\tt computeFromIntersection($k,i,next$)}\\
       
  }
  }
  }
  
  \hrule\medskip
  \SetKwBlock{Begin}{}{}

  function \texttt{computeFromRight($k,i,next$):} \hfill // compute the left ladder ${\cal L}^k[i]$
  \Begin{
    $k'\gets 1$ \hfill // row of first increment\\
    $i'\gets next-1$ \hfill // index of first increment\\
    \While{$k'<k$}{
      \If{$\neg${\tt isConnected}($\IG{k'}{i'}$)}{
        $k \gets k'-1$ \hfill // move the original walk..\\
        $i \gets i'+1$ \hfill // ..below-right disconnected graph,\\
        {\tt return} \hfill // abort function
      }
      $k' \gets k' + 1$; 
      $i' \gets i' - 1$\\
      $\G^{k'}[i'] \gets \G^{k'-1}[i'+1] \cap G_{i'}$ \hfill // ``increment'' the ladder
    }
  }
  \hrule\medskip
  \SetKwBlock{Begin}{}{}
  function \texttt{computeFromIntersection($k,i,next$):} 
  \Begin{
  \eIf{$i=next-1$}{$\IG{k}{i} \leftarrow \IG{k-1}{i} \cap G_{i+k-1}$}{
    $k' \gets k-next+i$ \hfill // row of increment (right ladder)\\
    $\IG{k'}{next} \gets \IG{k'-1}{next} \cap G_{next+k'-1}$ \hfill // ``increment'' right ladder\\
    \eIf{$\neg${\tt isConnected}($\IG{k'}{next}$)}{$i \gets next$; $k \gets k'$}{
    $\IG{k}{i} \gets \IG{next-i}{i} \cap \IG{k'}{next}$  \hfill // compute intersection based on Lemma~\ref{lem:intersection}}
    }
   }

 \caption{\label{algo:optimal_interval_window} Optimal algorithm for \stable}
\end{algorithm}

The walk begins by constructing a right ladder. In each step, if a computed intersection graph $\IG{k}{j}$ is connected, and $G_i$, $i=j+k-1$, is the most recently received graph, then the value $k$ is returned to indicate that $G_{i}$ is $k$-stable. Then the walk climbs one row in the hierarchy to $\IG{k+1}{j}$ which takes into consideration the next graph $G_{i+1}$. If a computed intersection graph $\IG{k}{j}$, $k > 1$, is disconnected, then the walk descends to the next graph in the row below, i.e. to $\IG{k-1}{j+1}$. In this case no value is returned because the next graph in $\G$ has not yet been considered. If a graph $\IG{1}{j}$ is disconnected, then $0$ is returned, and the walk moves right to the next graph.

As in the previous algorithms, the right ladders are constructed incrementally as the walk goes up, even though each graph $\IG{k}{j}$ can be computed from $\IG{k-1}{j} \cap G_{j+k-1}$, because this prepares the ladders needed to compute the intersection graphs if the walk goes down. This is illustrated by the second right ladder $\rlad{5}{7}$ in Figure~\ref{fig:window_walk}. If a disconnected graph $\IG{k}{j}$ is found while building a right ladder, the walk jumps to the next graph in the row just below, i.e. to $\IG{k-1}{j+1}$, to avoid unnecessary computations. For example, in Figure~\ref{fig:window_walk} the walk jumps from $\IG{5}{7}$ which is disconnected to $\IG{4}{8}$ without computing $\IG{7}{5}$ and $\IG{6}{6}$.

If a graph $\IG{k}{j}$ cannot be computed using the current ladders, then a complete new left ladder $\llad{k}{k+j-1}$ is constructed as high as possible until it reaches a previously computed graph or until it encounters a disconnected graph. The former case is illustrated by the left ladder $\llad{5}{6}$ in Figure~\ref{fig:window_walk} which is built when the walk descends from $\IG{6}{1}$ to $\IG{5}{2}$. The latter case is illustrated by the left ladder $\llad{4}{11}$ which encounters the disconnected graph $\IG{4}{8}$. In this case the walk resumes from the previous graph in the ladder ($\IG{3}{9}$ in the example). In contrast, a new left ladder is not needed when the walk descends three times from $\IG{8}{2}$ to $\IG{5}{5}$ because the ladders $\llad{5}{6}$ and $\rlad{3}{7}$ that exist at this point can be used to compute these intersections.

In the example in Figure~\ref{fig:window_walk}, the sequence of values $T_1,T_2,T_3,\ldots,T_{14}$ that the algorithm outputs is $1,2,3,4,5,5,6,7,5,6,3,4,5,6$.

\begin{theorem}\label{thm:stable}
\stable can be solved with an amortized cost of $\Theta(1)$ elementary operations per graph received.
\end{theorem}
\begin{proof}
\review{
The complexity analysis of the algorithm is similar to that of the online algorithm for {\sc \inter}.
The number of intersection graphs in right ladders never exceeds the number of graphs received and the same is true for left ladders. Each intersection graph in a ladder is computed using one binary intersection operation. Each time that the walk climbs in the intersection hierarchy, one connectivity test is performed and a single graph is processed. When the walk descends, no new graph in $\G$ is processed, but the number of descents cannot exceed the number of ascents, and each descent uses at most one connectivity test. This results in a constant amortized cost for each received graph. \qed
}
\end{proof}
\section{Conclusions}

In this paper, we studied the problem of testing whether a given dynamic graph $\G=(G_1,G_2,...,G_{\delta})$ is $T$-interval connected. 
We also considered the related problem of finding the largest $T$ for which a given $\G$ is $T$-interval connected. 
We assumed that the dynamic graph $\G$ is a sequence of {\em independent} graphs and we investigated algorithmic solutions that use two elementary operations, {\em binary intersection} and {\em connectivity testing}, to solve the problems. We developed efficient algorithms that use only $O(\delta)$ elementary operations, asymptotically matching the lower bound of $\Omega(\delta)$. We presented PRAM algorithms that show that both problems can be solved efficiently in parallel, and online algorithms that use $\Theta(1)$ elementary operations per graph received. We also presented dynamic versions of the online algorithms that report connectivity based on recent network history.

In our study, we focused on algorithms using only the two elementary
operations {\em binary intersection} and {\em connectivity testing}.
This approach is suitable for a high-level study of these problems
when the details of changes between successive graphs in a sequence
are arbitrary.
If the evolution of the dynamic graph is constrained in some
ways (e.g., bounded number of changes between graphs), then one
could benefit from the use of more sophisticated data structures to
lower the complexity of the problem. Another natural extension of our
investigation of $T$-interval connectivity would be a similar study for
other classes of dynamic graphs, as identified in \cite{CFQS12}. \review{In particular, we think that the framework presented here could be generalized to testing other properties of subsequences of dynamic graphs.}

Distributed algorithms for all of these problems, in which a node in the
graph only sees its local neighbourhood, would also be of interest. For example, distributed versions of the dynamic algorithms in Section~\ref{sec:dynamic} could be used to supplement the information available to distributed Internet routing protocols such as OSPF (Open-Shortest Path First) which are used to construct routing tables. Our dynamic algorithms have $\Theta(1)$ amortized complexity, and distributed versions with $\Theta(1)$ amortized complexity could provide real-time information about network connectivity to OSPF.


\end{document}